\documentclass[aps,twocolumn,groupedaddress,superscriptaddress,amsmath,amssymb,amsthm]{revtex4}
\usepackage{subfigure,hyperref,bbm,times}
\usepackage[T1]{fontenc}
\usepackage{braket}
\usepackage{amsthm}
\usepackage{epsfig}
\usepackage{color}
\usepackage{graphicx}
\usepackage{dcolumn}
\usepackage{bm}
\newtheorem{theor}{Theorem}
\newtheorem{lem}{Lemma}
\newcommand{\ave}[1]{\langle#1\rangle}

\providecommand{\openone}{\leavevmode\hbox{\small1\kern-3.8pt\normalsize1}}

\begin{document}

\title{Nonlocality threshold for entanglement under general dephasing evolutions: A case study}

\author{Rosario Lo Franco}
\affiliation{Dipartimento di Energia, Ingegneria dell'Informazione e Modelli Matematici, Universit\`{a} di Palermo, Viale delle Scienze, Ed. 9, 90128 Palermo, Italy}
\email{rosario.lofranco@unipa.it}

\begin{abstract}
Determining relationships between different types of quantum correlations in open composite quantum systems is important since it enables the exploitation of a type by knowing the amount of another type. We here review, by giving a formal demonstration, a closed formula of the Bell function, witnessing nonlocality, as a function of the concurrence, quantifying entanglement, valid for a system of two noninteracting qubits initially prepared in extended Werner-like states undergoing any local pure-dephasing evolution. This formula allows for finding nonlocality thresholds for the concurrence depending only on the purity of the initial state. We then utilize these thresholds in a paradigmatic system where the two qubits are locally affected by a quantum environment with an Ohmic class spectrum. We show that steady entanglement can be achieved and provide the lower bound of initial state purity such that this stationary entanglement is above the nonlocality threshold thus guaranteeing the maintenance of nonlocal correlations.     
\end{abstract}

\maketitle

\section{Introduction}
Quantum nonlocality and entanglement identify two kinds of quantum correlations that are at the foundations of quantum mechanics and play an essential role in quantum information theory \cite{amico2008RMP,horodecki2009RMP,brunnerRMP}. Although for pure states of a bipartite system the two concepts coincide, for mixed states the presence of entanglement does not imply in general nonlocality, the latter being a stronger quantumness property than entanglement \cite{brunnerRMP,werner}. In many scenarios a quantum system exists in mixed states, particularly during its dynamics due to the interaction with the surrounding environment which induces decoherence and therefore mixedness \cite{petru,rivasreview}. Moreover, the environmental effects are usually detrimental for entanglement and nonlocality in a configuration of separated noninteracting qubits embedded in independent local environments \cite{lofrancoreview,aolitareview}. Such a situation is typical for quantum communication and information protocols where the supports of the single-particle wave functions are centered around sufficiently far locations or the particles are separated by an energy barrier large enough, as in experiments with photons in different optical modes or with strongly repelling trapped ions \cite{obrienreview,cirac2001PRA,xu2013,darrigoAOP,orieux2015}. It is thus useful to find dynamical equations that relate these two kinds of quantum correlations, for instance in order to acquire information on the presence of nonlocality just by looking at the quantum entanglement of the overall system. 

The possible existence of closed relations between quantifiers of entanglement and nonlocality is a subject of special interest in dynamical contexts \cite{mazzolapalermo2010PRA,horst,Bartkiewicz}. Entanglement can be quantified for an arbitrary two-qubit state by the concurrence $C(t)$ \cite{Wootters98}. 
Nonlocality is instead unambiguously identified if a Bell inequality is violated. The Bell function $\mathcal{B}$, as defined by the Clauser-Horne-Shimony-Holt (CHSH) form of the Bell inequality \cite{brunnerRMP}, can be then used to seek whether the system exhibits nonlocal correlations which occur with certainty if $\mathcal{B}>2$.
An important goal within this context is to prove the existence of threshold values of concurrence which ensure nonlocal quantum correlations during a given evolutions whenever $C(t)$ stays above that threshold. In fact, apart its basic interest, this behavior would permit to utilize the system state for quantum information processes relying on nonlocality, as device-independent and security-proof quantum key distribution protocols \cite{brunnerRMP,gisinnatphot}.
It has been recently reported that Bell function and concurrence satisfy a closed dynamical relation in the case of general pure-dephasing evolution for each qubit, provided that the qubits are initially prepared in the class of extended Werner-like states \cite{lofrancoPRB}. This result has been then utilized to verify the preservation of nonlocality by dynamical decoupling pulse sequences in the presence of low-frequency noise typical of the solid state \cite{lofrancoPRB}. 

The aim of this work is to review the pure-dephasing relation between concurrence and Bell function, giving a formal demonstration of it. Successively, we apply the result to a paradigmatic system made of two independent qubits locally interacting with pure-dephasing noise sources having Ohmic class spectrum. By exploiting the trapping of the single-qubit coherence which occurs in such a system \cite{addisPRA}, we investigate the possibility to achieve a stationary entanglement which ensures the presence of nonlocal correlations. In particular, we find the lower bound of the purity of the initial two-qubit state such as to guarantee that the stationary concurrence is above the nonlocality threshold.

\section{Preliminaries: Bell function and concurrence}
In order to witness nonlocality we use here the CHSH form of Bell inequality, which is the most suitable for an experimental test of nonlocality in bipartite quantum systems composed by two-level (spin-like) systems, or qubits \cite{brunnerRMP}. Let the operator $\mathcal{O}_S=\mathcal{O}_S(\theta_S,\phi_S)$ be a spin observable with eigenvalues $\pm1$ associated to the spin particle $S=A,B$ defined by $\mathcal{O}_S=\textbf{O}_S\cdot\bm{\sigma}_S$, where
$\textbf{O}_S\equiv(\sin\theta_S\cos\phi_S,\sin\theta_S\sin\phi_S,\cos\theta_S)$ is the unit vector indicating an arbitrary direction in the spin space and $\bm{\sigma}_S=(\sigma_x^S,\sigma_y^S,\sigma_z^S)$ is the Pauli matrices vector. The CHSH-Bell inequality associated to a two-qubit state $\rho$ is  $\mathcal{B}(\rho)\leq2$, where $\mathcal{B}(\rho)$ is the Bell function defined as \cite{brunnerRMP,bellomo2010PLA}
\begin{equation}\label{CHSHBellinequality}
\mathcal{B}(\rho)=|\ave{\mathcal{O}_A\mathcal{O}_B}+\ave{\mathcal{O}_A\mathcal{O}'_B}
+\ave{\mathcal{O}'_A\mathcal{O}_B}-\ave{\mathcal{O}'_A\mathcal{O}'_B}|,
\end{equation}
where $\ave{\mathcal{O}_A\mathcal{O}_B}=\mathrm{Tr}\{\hat{\rho}\mathcal{O}_A\mathcal{O}_B\}$ is the correlation function of observables $\mathcal{O}_A$, $\mathcal{O}_B$ and $\mathcal{O}'_S\equiv\mathcal{O}_S(\theta'_S,\phi'_S)$. If, given the state $\rho$, it is possible to find a set of angles $\{\theta_A,\theta'_A,\theta_B,\theta'_B\}$ and $\{\phi_A,\phi'_A,\phi_B,\phi'_B\}$ such that the CHSH-Bell inequality is violated ($\mathcal{B}(\rho)>2$), then the correlations are nonlocal or, in other words, the system state exhibits nonlocality.
A procedure to obtain the maximum of $\mathcal{B}(\rho)$ and the corresponding angles for an arbitrary two-spin-$1/2$ mixed state is well known \cite{horodecki1995PLA}. Using it, one finds $\mathcal{B}_\mathrm{max}(\rho)=2\sqrt{\mathrm{max}_{j>k}\{u_j+u_k\}}$, where $j,k=1,2,3$ and $u_j$'s are three quantities depending on the state.

Entanglement for an arbitrary state $\rho$ of two qubits is quantified by concurrence \cite{amico2008RMP,Wootters98}
\begin{equation}
\mathcal{C}_{AB}=\mathcal{C}(\rho)=\textrm{max}\{0,\sqrt{\lambda_{1}}-\sqrt{\lambda_{2}}-\sqrt{\lambda_{3}}-\sqrt{\lambda_{4}}\},
\end{equation} 
where $\lambda_{i}$ ($i=1,\ldots,4$) are the eigenvalues in decreasing order of the matrix $\rho(\sigma_{y}\otimes\sigma_{y})\rho^{\ast}(\sigma_{y}\otimes\sigma_{y})$, with $\sigma_{y}$ denoting the second Pauli matrix and $\rho^{\ast}$ corresponding to the complex conjugate of the two-qubit density matrix $\rho$ in the standard computational basis $\mathcal{B}=\{\ket{1}\equiv\ket{11},\ket{2}\equiv\ket{10},\ket{3}\equiv\ket{01},\ket{4}\equiv\ket{00}\}$.

Let us take now the two-qubit states whose density matrix, in the basis $\mathcal{B}$, has a X structure as
\begin{equation}\label{Xstatesdensitymatrix}
   \hat{\rho}_X = \left(
\begin{array}{cccc}
  \rho_{11} & 0 & 0 & \rho_{14}  \\
  0 & \rho_{22} & \rho_{23} & 0 \\
  0 & \rho_{23}^* & \rho_{33} & 0 \\
  \rho_{14}^* & 0 & 0 & \rho_{44} \\
\end{array}
\right).
\end{equation}
This class of states is sufficiently general to include Bell states (pure two-qubit maximally entangled states) and Werner states (mixture of Bell states with white noise) \cite{bellomo2008PRA}. A remarkable aspect of the X states is that, under various kinds of dynamics, the initial X structure is maintained during the evolution \cite{bellomo2008PRA,bellomo2007PRL}. Using the above criterion to obtain $\mathcal{B}_{\mathrm{max}}$ \cite{horodecki1995PLA}, the three quantities $u$'s, in terms of the density matrix elements, are found to be
\begin{eqnarray}\label{uXstate}
u_1&=&4(|\rho_{14}|+|\rho_{23}|)^2,\quad
u_2=(\rho_{11}+\rho_{44}-\rho_{22}-\rho_{33})^2,\nonumber\\ 
u_3&=&4(|\rho_{14}|-|\rho_{23}|)^2,
\end{eqnarray}
as already reported \cite{derkacz2005PRA}. Being $u_1$ always larger than $u_3$, the maximum of Bell function for X states results to be $\mathcal{B}_\mathrm{max}(\hat{\rho})=2\sqrt{u_1+\mathrm{max}_{j=2,3}\{u_j\}}$, that is
\begin{eqnarray}\label{BB1B2}
\mathcal{B}_\mathrm{max}&=&\mathrm{max}\{\mathcal{B}_1,\mathcal{B}_2\},\quad 
\mathcal{B}_1=2\sqrt{u_1+u_2},\nonumber\\  \mathcal{B}_2&=&2\sqrt{u_1+u_3}.
\end{eqnarray}

The expression of concurrence for X states is given by \cite{lofrancoreview}
\begin{eqnarray}\label{concxstate}
\mathcal{C}_\rho^X(t)&=&2\mathrm{max}\{0,|\rho_{23}(t)|-\sqrt{\rho_{11}(t)\rho_{44}(t)},\nonumber\\
&&|\rho_{14}(t)|-\sqrt{\rho_{22}(t)\rho_{33}(t)}\}.
\end{eqnarray}

\section{Relation between Bell function and concurrence under pure-dephasing}
In this Section we review the relation between quantifiers of entanglement and nonlocality for two noninteracting qubits, $A$ and $B$, locally subject to any pure-dephasing interaction with the environment. We start by defining the general form of the Hamiltonian and the relevant initial states of the two-qubit system. 

The dynamics of each qubit is governed by the pure-dephasing Hamiltonian ($\hbar=1$)
\begin{equation}\label{pure-dephasingH}
H= {\Omega \over 2}\sigma_z  + {\hat{\chi} \over 2} \sigma_z + \hat{H}_R,
\end{equation}
where $\Omega$ is the Bohr frequency of the qubit and $\hat{\chi}$ represents an arbitrary environmental operator
coupled to the same qubit. The free evolution of the environment is included in $\hat{H}_R$. The overall Hamiltonian of the two-qubit system is 
thus $H_\mathrm{tot}=H_A+H_B$, with each $H_S$ ($S=A,B$) given by Eq.~(\ref{pure-dephasingH}). 

The two qubits are supposed to be prepared in an EWL state \cite{bellomo2008PRA}
\begin{equation}\label{EWLstates}
   \rho_1=r \ket{1_{a}}\bra{1_{a}}+\frac{1-r}{4}\openone_4,\quad
    \rho_2=r \ket{2_{a}}\bra{2_{a}}+\frac{1-r}{4}\openone_4,
\end{equation}
where the pure parts $\ket{1_{a}}=a\ket{01}+b\ket{10}$ and $\ket{2_{a}}=a\ket{00}+b\ket{11}$ are, respectively, 
the one-excitation and two-excitation Bell-like states with $a^2+|b|^2=1$ ($a$ is real) and $\openone_4$ is the $4\times4$ identity matrix of the two-qubit Hilbert space. The EWL states are a subclass of the X states of Eq. (\ref{Xstatesdensitymatrix}) and when $a=b=1/\sqrt{2}$ they reduce to
the Werner states, a subclass of Bell-diagonal states \cite{horodecki2009RMP}. The purity parameter $r\in[0,1]$ is a measure of the purity of EWL states which is given by $P=\mathrm{Tr}(\rho^2)=(1+3r^2)/4$. 
The two EWL states of Eq.~(\ref{EWLstates}) have the same value of concurrence 
$\mathcal{C}_{\rho_1}(0)=C_{\rho_2}(0)=2\mathrm{max}\{0,(|ab|+1/4)r-1/4\}$. Initial states are thus entangled for $r>\bar{r}=(1+4|ab|)^{-1}$. 

We can now claim the following theorem.

\begin{theor} \label{TH1}
For noninteracting qubits starting from an entangled EWL state ($r>\bar{r}$) and locally subject to any pure-dephasing noise, the Bell function and the concurrence satisfy the dynamical relation \cite{lofrancoPRB}
\begin{equation}\label{BversusC}
\mathcal{B}(t)=2\sqrt{r^2+[\mathcal{C}(t)+(1-r)/2]^2}.
\end{equation}
\end{theor}

\begin{proof} 
Since the two qubits are noninteracting, the evolution of entanglement and nonlocality can be simply obtained from the
knowledge of single-qubit dynamics \cite{bellomo2007PRL}.
Under a pure-dephasing evolution, for each qubit the diagonal elements of the density matrix in the eigenstate basis
remain unchanged. The single-qubit coherences evolve in time as $q_S(t)\equiv\rho_{01}^S(t)/\rho_{01}^S(0)$, 
the explicit time dependence being specified by the environmental properties and the interaction term. 
If the system is subject to pure-dephasing only,  the X form of 
the density matrix is maintained during the evolution. 
In particular, diagonal elements remain constant while 
off-diagonal elements evolve in time. They 
are related to the single qubit coherences by  \cite{lofrancoreview,bellomo2007PRL}
$\rho_{23}(t)=\rho_{23}(0) q_A(t) q_B^\ast(t)$ for the initial state $\rho_1$ and 
$\rho_{14}(t)=\rho_{14}(0) q_A(t) q_B(t)$ for $\rho_2$.
From Eq.~(\ref{concxstate}), we immediately obtain the concurrences at time $t$ for the two initial states given in Eq.~(\ref{EWLstates}) as
$\mathcal{C}_{\rho_1}(t)=2\mathrm{max}\{0,|\rho_{23}(t)|-\sqrt{\rho_{11}(0)\rho_{44}(0)}\}$ and 
$\mathcal{C}_{\rho_2}(t)=2\mathrm{max}\{0,|\rho_{14}(t)|-\sqrt{\rho_{22}(0)\rho_{33}(0)}\}$.
For the pure-dephasing evolution, we then have
$\mathcal{C}_{\rho_1}(t)=\mathcal{C}_{\rho_2}(t) \equiv \mathcal{C}(t)$ with
\begin{equation}\label{Kt}
\mathcal{C}(t)=2\mathrm{max}\{0,r a\sqrt{1-a^2}|q_A(t)q_B(t)|-(1-r)/4\}.
\end{equation}
We now turn to nonlocality. For independent qubits subject to local pure-dephasing noise, the two functions $\mathcal{B}_1$, $\mathcal{B}_2$ defined in Eq. (\ref{BB1B2}) have the same form for the initial EWL states of Eq.~(\ref{EWLstates}) and are given by
\begin{eqnarray}\label{B1B2}
&\mathcal{B}_1(t)=2\sqrt{r^2+4r^2a^2(1-a^2)|q_A(t) q_B(t)|^2},&\nonumber\\
&\mathcal{B}_2(t)=4\sqrt{2}r|a|\sqrt{1-a^2}|q_A(t) q_B(t)|.&
\end{eqnarray}
Since $a,r,|q_A(t)|,|q_B(t)|\in[0,1]$, it is easily seen that $\mathcal{B}_1(t)\geq \mathcal{B}_2(t)$ for any $t$: the maximum Bell function 
is thus $\mathcal{B}(t)=\mathcal{B}_1(t)$. 
We now observe that nonzero entanglement ($\mathcal{C}(t)>0$) is necessary in order that the state exhibits nonlocality and we can simply write $\mathcal{C}(t)=2[r|a|\sqrt{1-a^2}|q_A(t)q_B(t)|-(1-r)/4]$. From the first line of Eq.~(\ref{B1B2}) with $\mathcal{B}(t)=\mathcal{B}_1(t)$, it is promptly found that the second term under square root is equal to $[\mathcal{C}(t)+(1-r)/2]^2$ so that the result of Eq.~(\ref{BversusC}) is proven. 
\end{proof}

We stress that this theorem is valid for any local pure-dephasing qubit-environment interaction, starting from 
an initial EWL state with a generic value of $a\neq 0,1$.
For instance, when $r=1$ Eq.~(\ref{BversusC}) reduces to $\mathcal{B}(t)=2\sqrt{1+\mathcal{C}(t)^2}$: the latter relation, which is known to occur when the system is in a pure state \cite{gisin1991PLA} or in a Bell-diagonal state \cite{verstraete2002PRL}, is here found to hold during the system evolution for more general states. 

As a consequence of the previous theorem, we get the following
\begin{lem}
The threshold value $\mathcal{C}_\mathrm{th}$ of the concurrence such that, when $\mathcal{C}(t)>\mathcal{C}_\mathrm{th}$, the evolved two-qubit state $\rho(t)$ exhibits nonlocality at time $t$ under any pure-dephasing evolution and for EWL initial states is
\begin{equation}\label{Cthreshold}
\mathcal{C}_\mathrm{th}=\sqrt{1-r^2}-(1-r)/2.
\end{equation}
\end{lem}
\begin{proof} The demonstration immediately follows from Eq.~(\ref{BversusC}) by solving the tight inequality $\mathcal{B}(t)>2$ with respect to concurrence $\mathcal{C}(t)$. 
\end{proof}

We name $\mathcal{C}_\mathrm{th}$ nonlocality threshold. Therefore, for initial EWL states evolving under any pure-dephasing interaction, the system exhibits nonlocality at time $t$ provided that the concurrence $C(t)$ is larger than a threshold value $\mathcal{C}_\mathrm{th}$ depending only on the system initial purity. The threshold is a decreasing function of the purity and for $r=1$ it is $\mathcal{C}_\mathrm{th}=0$: this implies that if the two qubits are initially in a pure entangled Bell-like state, their correlations remain nonlocal until any amount of nonzero entanglement is present. In the following we exploit these results in a specific system where they are particularly relevant.

\section{Evolution under an environment with Ohmic class spectrum}
We now apply the previous results in the dynamics of entanglement between two initially correlated but noninteracting qubits $A$ and $B$, each undergoing local pure-dephasing due to the coupling with a zero-temperature bosonic environment having a Ohmic-like spectrum, whose characteristics do not depend on the qubit \cite{haikkaPRA}. In this case, the explicit form of the Hamiltonian of Eq.~(\ref{pure-dephasingH}) ruling the dynamics of each qubit is ($\hbar=1$)
\begin{equation}\label{HOhmic}
H=\omega_0\sigma_z + \sum_k\omega_k a_k^\dagger a_k + \sigma_z \sum_k \left(g_k a_k + g_k^*a_k^\dagger \right),
\end{equation}
where $\omega_0$ is the qubit frequency, $\omega_k$ the frequencies of the reservoir modes, $\sigma_z$ the Pauli operator along the $z$-direction, $a_k$ the bosonic annihilation operators, $a_k^\dagger$ the bosonic creation operators and $g_k$ the coupling constants between the qubit and each reservoir mode. 
It is readily seen that the above Hamiltonian of Eq.~(\ref{HOhmic}) is obtained by the general pure-dephasing Hamiltonian of Eq.~(\ref{pure-dephasingH}) by substituting $\Omega=2\omega_0$, $H_R=\sum_k\omega_k a_k^\dagger a_k$ and $\hat{\chi}=2\sum_k \left(g_k a_k + g_k^*a_k^\dagger \right)$.
In the continuum limit one has \cite{petru} $\sum_k {\left|g_k \right|}^2\rightarrow \int d\omega J(\omega)\delta(\omega_k - \omega)$, where $J(\omega)$ is the reservoir spectral density which in this case is given by
\begin{equation}\label{spectrum}
J(\omega)=\frac{\omega^s}{\omega_c^{s-1}}e^{-\omega/\omega_c},
\end{equation}
with $\omega_c$ denoting the cut-off reservoir frequency ($\omega_c$ is assumed equal for the two independent environments). This qubit-environment system is exactly solvable in the case of zero temperature \cite{addisPRA,haikkaPRA}. In particular, the single-qubit decay characteristic function $q(t)=\rho_{01}(t)/\rho_{01}(0)$ at zero temperature is given by $q(t)=e^{-\Lambda(t)}$, with dephasing factor 
\begin{equation}\label{Lambda}
\Lambda(t)=2\int_0^t\gamma(t')dt',
\end{equation}
where $\gamma(t)$ is the dephasing rate of the single-qubit evolution \cite{haikkaPRA}
\begin{equation}\label{gamma}
\gamma(t)=\omega_c {\left[1+{\left(\omega_c t \right)}^2  \right] }^{-s/2}\Gamma[s]\sin\left[s \arctan(\omega_c t ) \right],
\end{equation}
with $\Gamma[x]$ being the Euler Gamma function. The parameter $s$ of the spectrum of Eq.~(\ref{spectrum}) governs the character of the dynamics: in particular, for $0\leq s \leq1$ the dynamics is Markovian (memoryless) while for $s>1$ it is non-Markovian (with memory), as indicated by the time behavior of the single-qubit coherence $q(t)=e^{-\Lambda(t)}$ \cite{addisPRA,haikkaPRA}. In fact, for $0\leq s \leq1$ the coherence vanishes asymptotically while for $s>1$ it reaches a stationary value which is not consistent with a simple Markovian approximation. Moreover, for $s>2$ an information backflow from the environment to the qubit occurs \cite{addisPRA}. 

If the two-qubit are initially prepared in an EWL state of Eq.~(\ref{EWLstates}), then the concurrence evolves as in Eq.~(\ref{Kt}) by substituting the single-qubit decay characteristic functions $q_S(t)=e^{-\Lambda_S(t)}$ ($S=A,B$), where each $\Lambda_S(t)$ is defined in Eq.~(\ref{Lambda}).
In general, the single-qubit decay function $q(t)$ can be different for the two qubits, $q_A(t)$, $q_B(t)$ depending on the values of the parameter $s_A$, $s_B$, respectively. In the following we shall study the case of both qubits locally subject to identical noise (that is, $q_A(t)=q_B(t)=q(t)$) and the situation where a qubit, for instance $B$, is isolated so that $q_B(t)=1$ at any time.

A first analysis of the entanglement dynamics for some values of the parameters is displayed in Fig.~\ref{fig:concevo}. In particular, for both qubits under dephasing we observe that for purity parameter $r=0.9$ and weights $a=1/\sqrt{2}$ (left panel), the entanglement disappears at a finite time, with the faster decay occurring for the larger values of the ohmicity parameter $s$ which indicate stronger single-qubit non-Markovianity. Therefore, under these initial conditions non-Markovian features do not help maintaining entanglement but are detrimental. On the other hand, for larger values of the purity of the initial state, stationary nonzero entanglement is found, as displayed in the right panel of Fig.~\ref{fig:concevo} for $r=1$. We point out how the case where a qubit is isolated is advantageous for the entanglement preservation, as expected. Notice also that, for $r=1$, the nonlocality threshold for the concurrence obtained from Eq.~(\ref{Cthreshold}) is zero, so that this stationary entanglement ensures the presence of nonlocal correlations. 

\begin{figure}[tbp]
\begin{center}
{\includegraphics[width=0.46\textwidth]{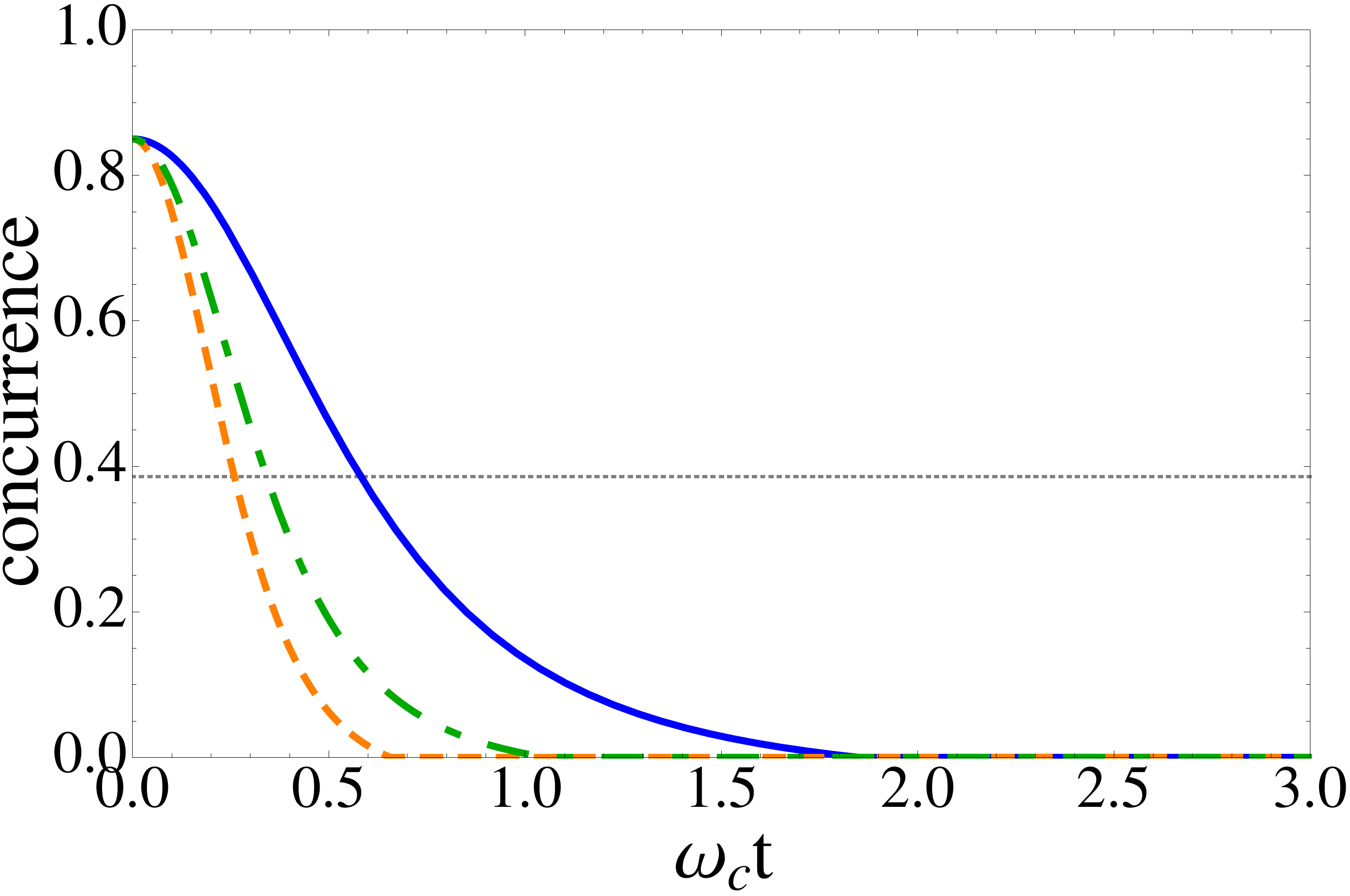}\hspace{0.2 cm}
\includegraphics[width=0.46\textwidth]{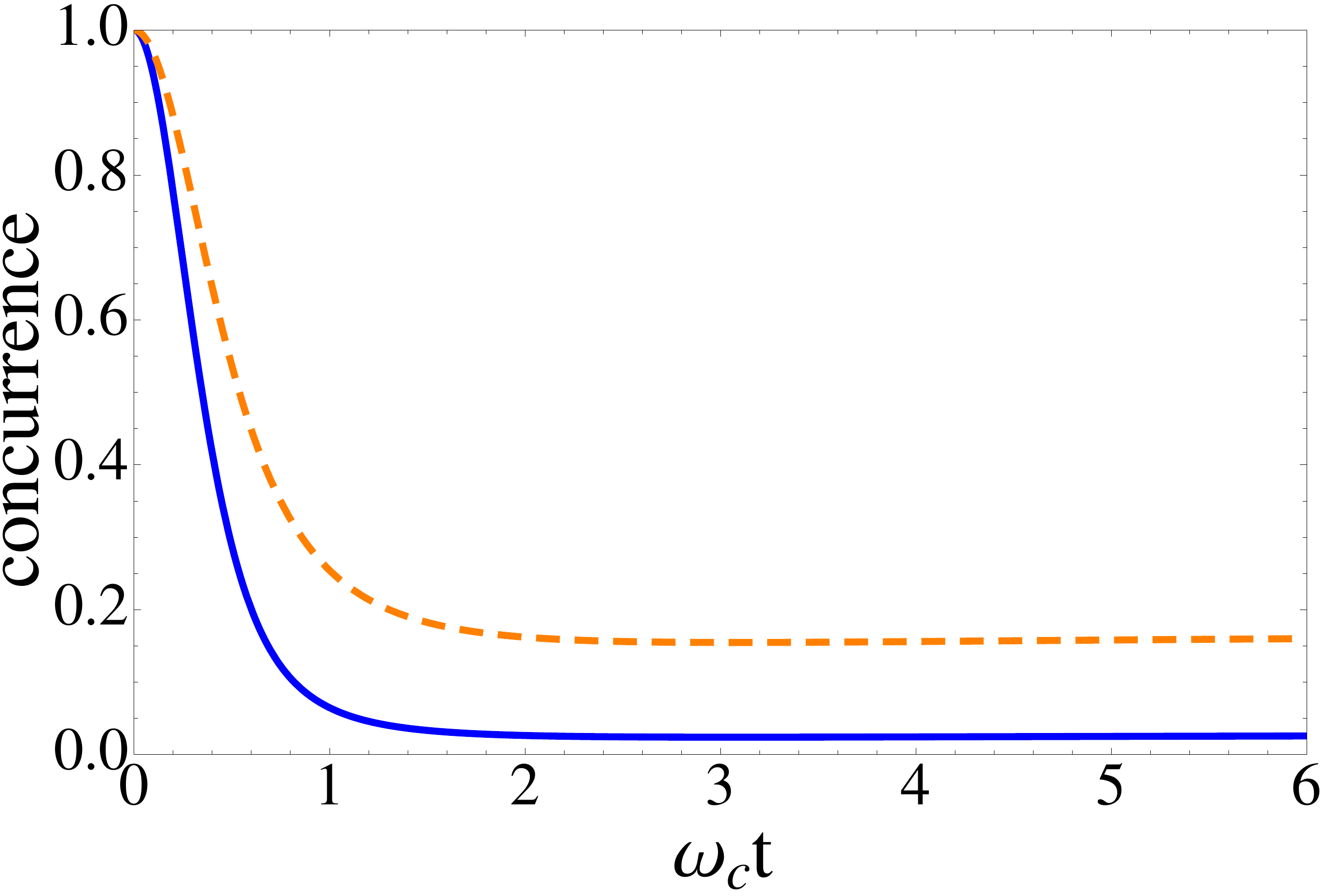}}
\end{center}
\caption{\footnotesize Dynamics of the two-qubit concurrence starting from Werner states ($a=b=1/\sqrt{2}$), for different parameter conditions. 
\textbf{Upper panel:} The purity parameter is $r=0.9$ for both the three curves and $s_A=s_B=1.5$ (solid blue line), $s_A=s_B=3$ (dashed orange line), $s_A=1.5$, $s_B=3$ (green dot-dashed line). The gray dotted line represents the nonlocality threshold for the concurrence ($\mathcal{C}_\mathrm{th}=0.386$).    
\textbf{Lower panel:} The purity parameter is $r=1$ and $s_A=s_B=2.5$ (solid blue line), $s_A=2.5$ while qubit $B$ is isolated (dashed orange line). In this case the nonlocality threshold is zero and entanglement trapping is found.}
\label{fig:concevo}
\end{figure}

The emergence of this stationary entanglement, which did not surface in a previous analysis of the same model due to the initial conditions there considered \cite{cianciarusoJPA}, strictly stems from the trapping of the single-qubit coherence \cite{addisPRA}. In fact, an analytic expression of the dephasing factor of Eq.~(\ref{Lambda}) can be obtained by solving the integral which gives
\begin{eqnarray}
\Lambda(\tau)&=&\frac{2 \Gamma (s)}{s-1}\left\{1-\left(\tau ^2+1\right)^{-\frac{s}{2}} 
\left[\tau  \sin \left(s \arctan\tau \right)\right. \right. \nonumber \\
&&\left.\left. +\cos \left(s \arctan\tau \right]\right)\right\},
\end{eqnarray}
where $\tau = \omega_c t$ is the dimensionless time in units of $\omega_c^{-1}$. It is then straightforward to see that, in the limit of $\tau$ to infinity, one has
\begin{equation}\label{Lambdainfty}
\Lambda^{\infty}(s)=\lim_{\tau\rightarrow\infty} \Lambda(\tau) = 
\left\{\begin{array}{ll} +\infty, & 0\leq s \leq1\\ \\ 2 \Gamma (s)/(s-1), & s>1 \end{array}\right.
\end{equation}
Therefore, the single-qubit coherence $q(t)=e^{-\Lambda(t)}$ decays to zero for values of $s\leq1$, while for any value of $s>1$ it achieves a stationary value $q^\infty (s) = \exp\{-\Lambda^\infty(s)\}$ depending on $s$. A plot of the asymptotic dephasing factor $\Lambda^{\infty}(s)$ as a function of $s$ is displayed in Fig.~\ref{fig:Lambda}.
It is seen that $\Lambda^{\infty}(s)$ is a convex function exhibiting its minimum value at $\bar{s}\simeq 2.46$ where $\Lambda^{\infty}(\bar{s})\simeq 1.77$. The value $\bar{s}$ is thus the value of the ohmicity parameter $s$ which maximizes the single-qubit coherence \cite{addisPRA}.   

\begin{figure}[tbp]
\begin{center}
{\includegraphics[width=0.5\textwidth]{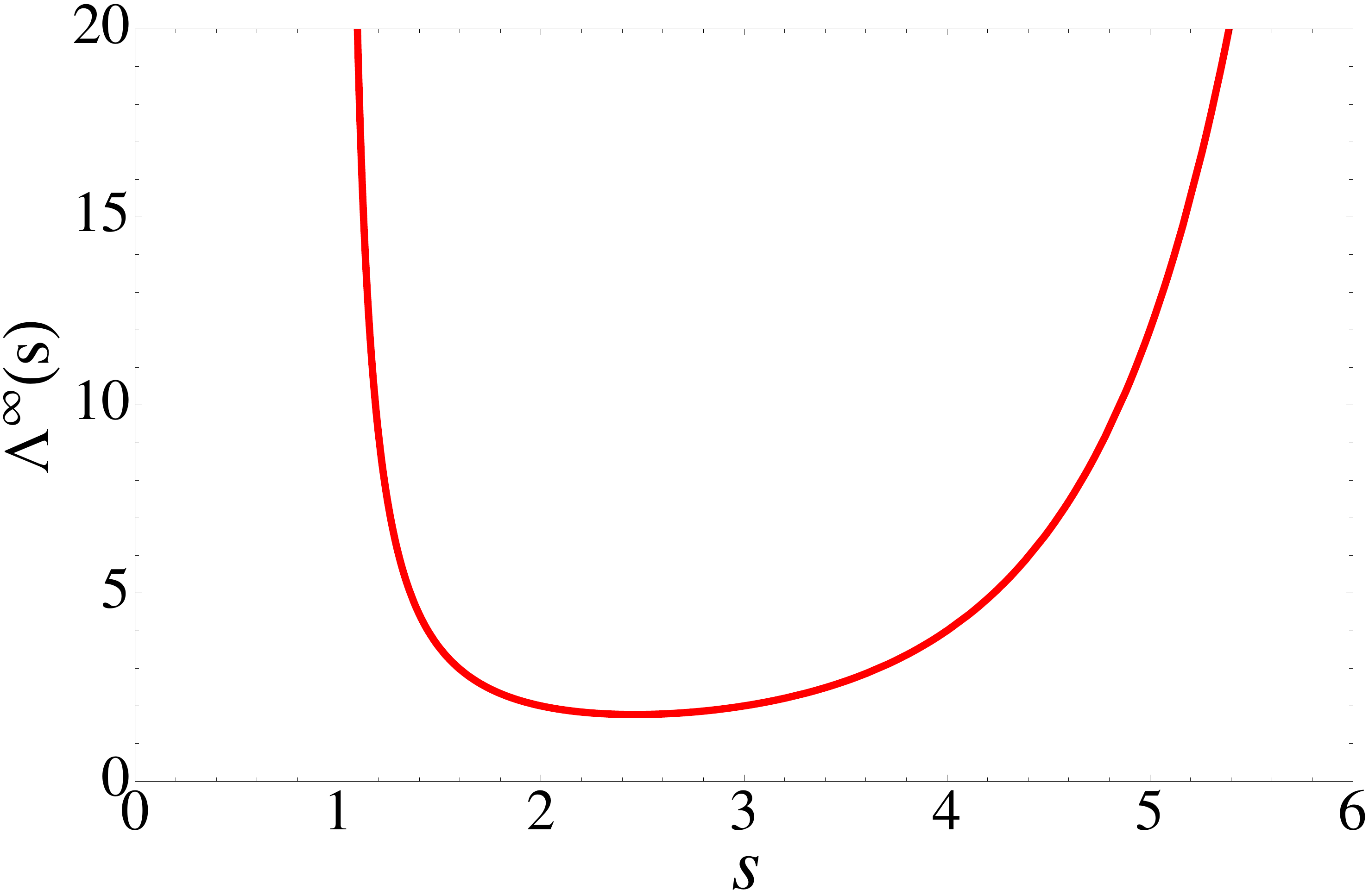}}
\end{center}
\caption{\footnotesize Asymptotic value, $\Lambda^{\infty}$, of the decay function $\Lambda(t)$ as a function of the spectrum parameter $s$. The minimum value of $\Lambda^{\infty}$ is attained for $\bar{s}\simeq 2.46$ which gives $\Lambda^{\infty}(\bar{s})\simeq 1.77$.}
\label{fig:Lambda}
\end{figure} 

However, the fact that for any $s>1$ a stationary coherence occurs does not entail that there is correspondingly a nonzero stationary entanglement. In fact, as we have already seen in Fig.~\ref{fig:concevo}, the occurrence of entanglement trapping strongly depends on the value of the purity $r$ of the initial EWL state. The analytic expression of the asymptotic concurrence is immediately obtained by Eq.~(\ref{Kt}) as
\begin{equation}\label{concasympt}
\mathcal{C}^\infty(s)=2\mathrm{max}\{0,ra\sqrt{1-a^2}\ e^{-n\Lambda^\infty(s)}-(1-r)/4\},
\end{equation}
where $\Lambda^\infty(s)$ is given in Eq.~(\ref{Lambdainfty}) and $n=1,2$ for the case when one qubit is isolated or both qubits are subject to noise, respectively. As seen from the above expression, the value of $\bar{s}$ which maximizes the stationary single-qubit coherence also maximizes the stationary concurrence. For $s\leq1$ one always get $\mathcal{C}^\infty(s)=0$. 
The values of $s>1$ such that $\mathcal{C}^\infty(s)>0$ for given $a$, $r$ can be easily obtained using Eq.~(\ref{concasympt}) and are determined by the inequality
\begin{equation}
\Lambda^{\infty}(s)=\frac{2 \Gamma (s)}{s-1}<\frac{1}{n}\ln[4ra\sqrt{1-a^2}/(1-r)].
\end{equation}
It is immediately seen that for $r=1$ and $a\neq0,1$ the right-hand side of the above inequality diverges to infinity and therefore the latter is satisfied for any value of $s>1$. Two plots of $\mathcal{C}^\infty(s)$ as a function of $s$ for $a=1/\sqrt{2}$ and two values of the purity parameter $r=1,0.99$ are shown in Fig~\ref{fig:concinfty}. For pure initial states such that $r=1$ (left panel) the nonlocality threshold is zero and $\mathcal{C}^\infty(s)>0$ for any value of $s>1$; for $r=0.99$ (right panel) the stationary value of concurrence is larger than zero for a finite range of $s$, but only for the case when a qubit is isolated this concurrence is able to overcome the nonlocality threshold $\mathcal{C}_\mathrm{th}=0.136$ for a smaller range of $s$ around $\bar{s}$.     

\begin{figure}[tbp]
\begin{center}
{\includegraphics[width=0.46\textwidth]{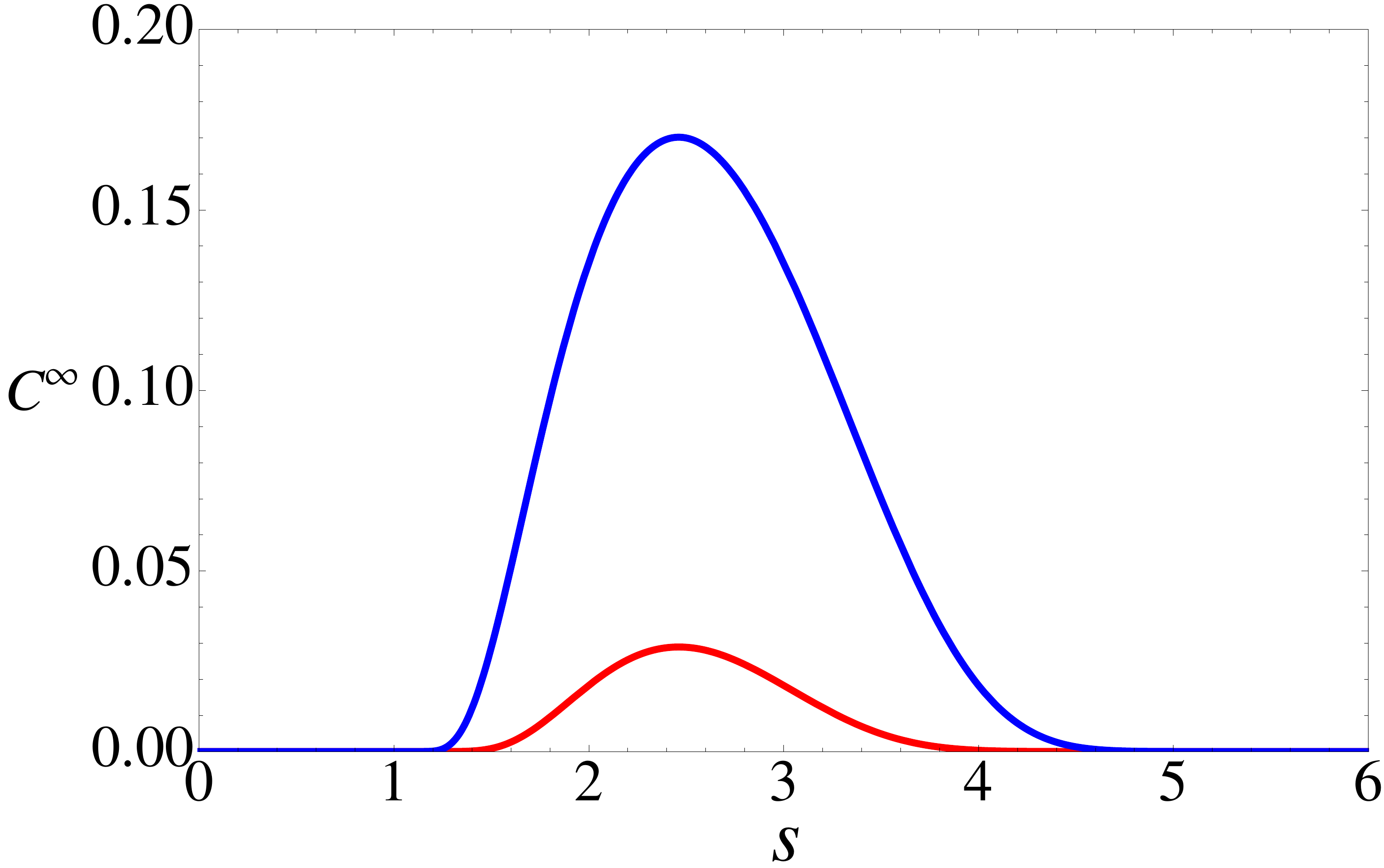}\hspace{0.2 cm}
\includegraphics[width=0.46\textwidth]{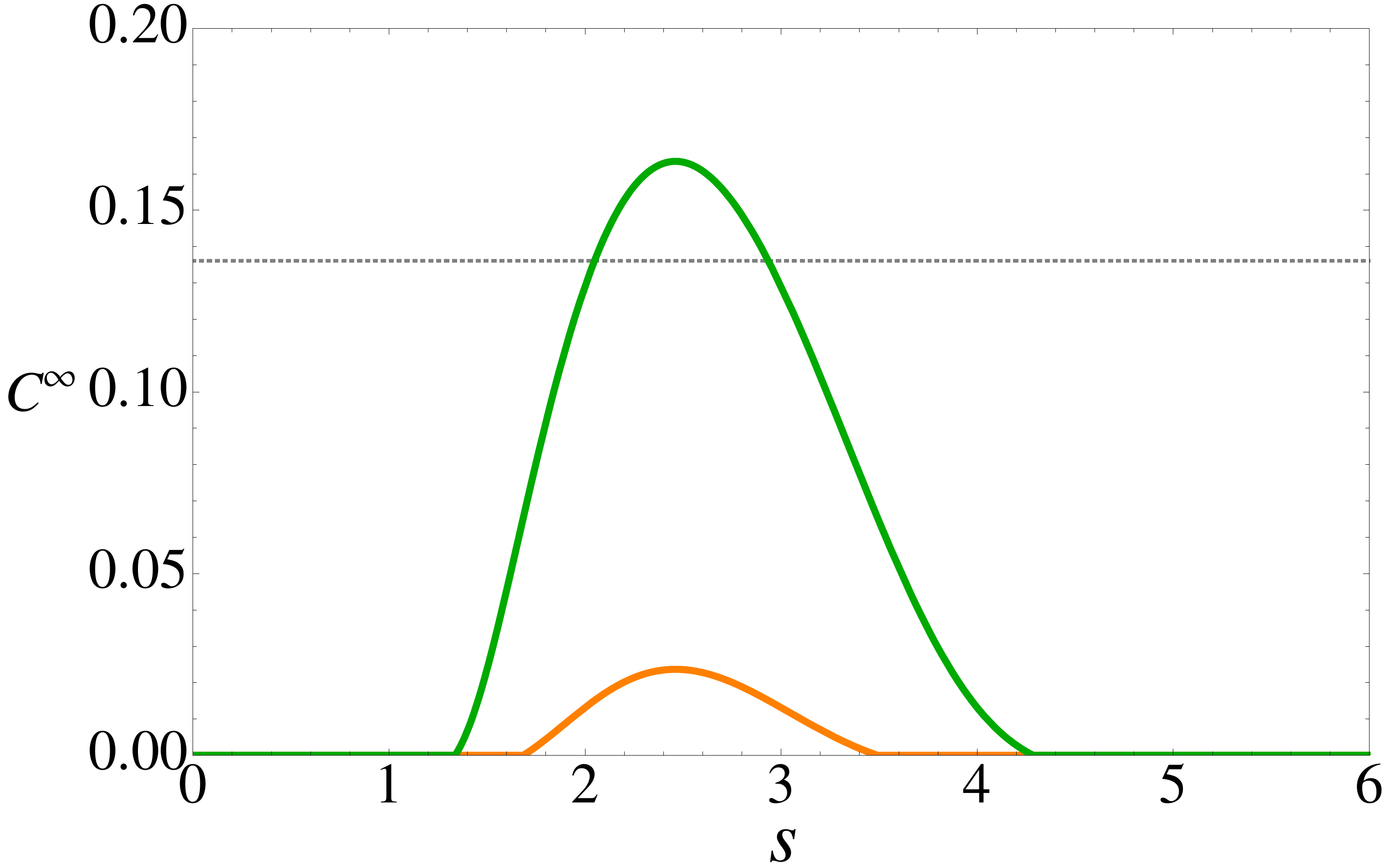}}
\end{center}
\caption{\footnotesize Asymptotic values of two-qubit concurrence starting from Werner states ($a=b=1/\sqrt{2}$) as a function of the ohmicity parameter $s$. In both panels the lower curve represents the concurrence for the case of both qubits subject to noise, while the upper curve is the concurrence for the case where a qubit is isolated. \textbf{Upper panel:} The purity parameter is $r=1$ and the nonlocality threshold is zero.  
\textbf{Lower panel:} The purity parameter is $r=0.99$ and the gray dotted line represents the nonlocality threshold for the concurrence ($\mathcal{C}_\mathrm{th}=0.136$). The upper curve, that is the stationary concurrence when a qubit is isolated, is larger than the nonlocality threshold within the range $2.06<s<2.94$.}
\label{fig:concinfty}
\end{figure}

So far, the results obtained tell us that, for $s>1$, larger values of the initial purity $r$ are more effective in order to get stationary entanglement which enables the presence of nonlocal correlations. We can make this qualitative behavior quantitative by finding the purity lower bound $r^\ast$ of the initial EWL states of Eq.~(\ref{EWLstates}) such that, once assigned the value of the weight $a$ and of the ohmicity parameter $s$, the stationary concurrence is above the nonlocality threshold for any $r>r^\ast$. This finding will permit us to know to which degree of purity the initial EWL entangled state must be prepared in order to guarantee that the two qubits will preserve nonlocality during the evolution. Exploiting the nonlocality threshold of Eq.~(\ref{Cthreshold}) together with the stationary concurrence of Eq.~(\ref{concasympt}) we can solve the inequality $\mathcal{C}^\infty>\mathcal{C}_\mathrm{th}$ for $s>1$ with respect to $r$ in order to obtain the nonlocality lower bound of purity
\begin{equation}\label{rast}
r^\ast=\sqrt{\frac{1}{1+4 \left(1-a^2\right) a^2 \exp\{-2n \Lambda^{\infty}(s)\}}},
\end{equation} 
where again $\Lambda^\infty(s)$ is given in Eq.~(\ref{Lambdainfty}) and $n=1,2$ corresponds to the case when one qubit is isolated or both qubits are subject to noise, respectively. From this equation, it is easily found that when $\Lambda^{\infty}(s)\rightarrow +\infty$ one has $r^\ast=1$, which means that for $s\leq1$ the initial entangled state must be pure in order to have nonlocal correlations preserved during the dynamics. In Fig.~\ref{fig:purity} the plots of $r^\ast$ for $n=1,2$ are displayed as a function of $s$ for $a=1/\sqrt{2}$ (left panel) and as a function of the weight $a^2$ for $s=\bar{s}$ (right panel). The plots show the sensitivity of the nonlocality purity lower bound $r^\ast$ to the system configuration, with very high values of initial purity of the EWL state required if both qubits are subject to noise. As expected, from Eq.~(\ref{rast}) it is promptly found that the minimum value of $r^\ast$ is obtained for $s=\bar{s}$ and $a=1/\sqrt{2}$ which are, respectively, the value of $s$ minimizing the stationary dephasing factor and the value of weight $a$ maximizing the entanglement of the pure part of the initial EWL state of Eq.~(\ref{EWLstates}). This result is confirmed by comparing the left and right panels of Fig.~\ref{fig:purity}. In particular, the minimum values of the nonlocality lower bound $r^\ast$ for $n=1,2$ are, respectively, $r^\ast_\mathrm{min}(n=1)\simeq 0.9858$ and $r^\ast_\mathrm{min}(n=2)\simeq 0.9996$. We point out that, despite these lower bounds of initial purities are very high, they are achievable for instance in all-optical setups with entangled polarized photons where the open quantum dynamics can be simulated by suitable quartz plates \cite{xu2013,orieux2015,chiuri2012}, in solid-state qubits (nuclear spins) in diamond \cite{pfaff2013} and in principle even in composite atom-cavity systems by purity swapping \cite{metwally2008}. Values of initial state purity smaller than the above lower bounds $r^\ast_\mathrm{min}$ will however lead to an entanglement evolution above the nonlocality threshold for long enough dimensionless times $\omega_c t$, as can be deducted from Fig.~1.

\begin{figure}[tbp]
\begin{center}
{\includegraphics[width=0.46\textwidth]{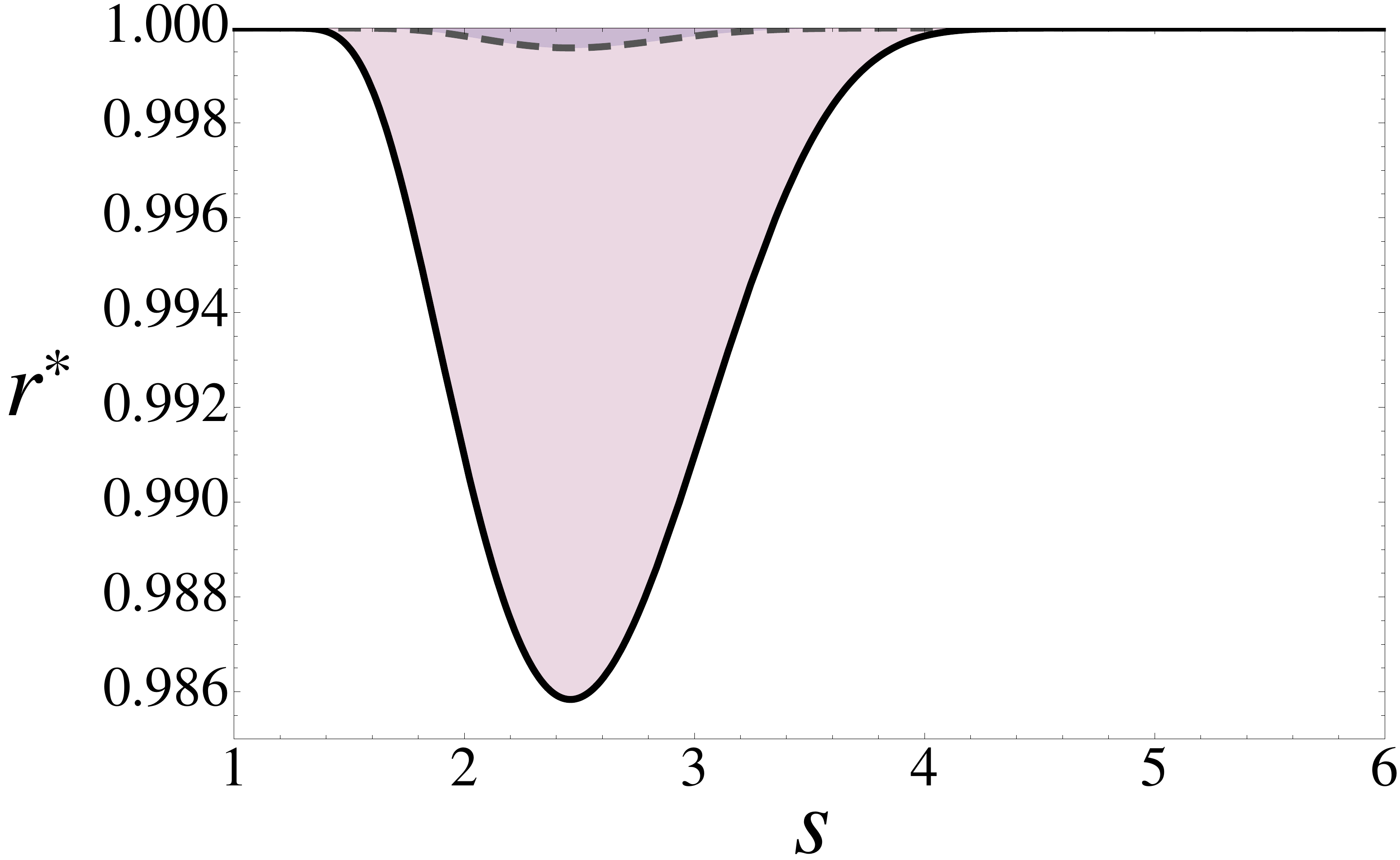}\hspace{0.2 cm}
\includegraphics[width=0.46\textwidth]{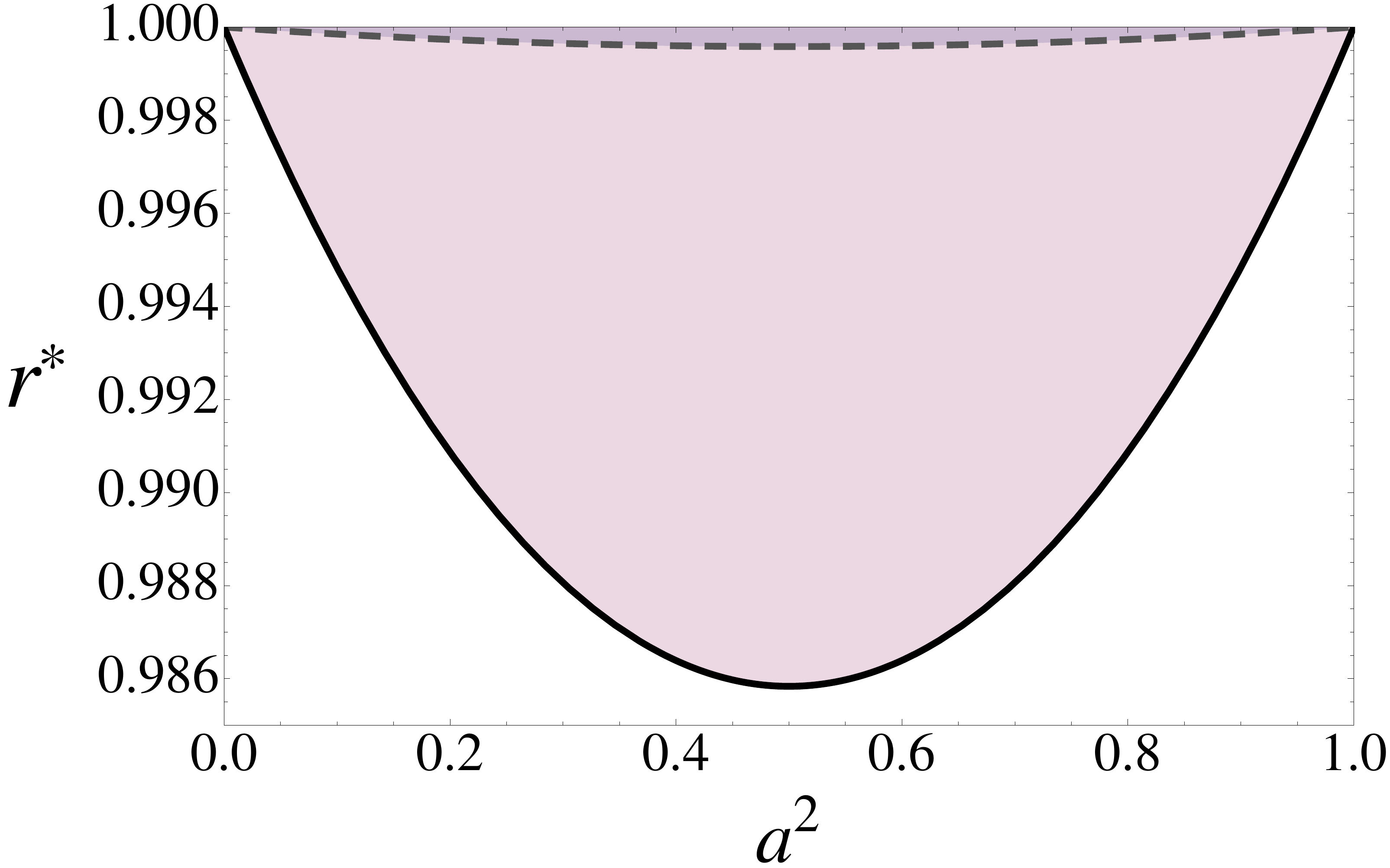}}
\end{center}
\caption{\footnotesize Lower bound $r^\ast$ of the purity parameter $r$, such that for any $r>r^\ast$ the asymptotic concurrence is larger than the nonlocality threshold ($\mathcal{C}^{\infty}>\mathcal{C}_\mathrm{th}$). \textbf{Upper panel:} $r^\ast$ as a function of the ohmicity parameter $s$ for $|a|=1/\sqrt{2}$. \textbf{Lower panel:} $r^\ast$ as a function of the weight $a^2$ for ohmicity parameter $s=\bar{s}\simeq2.46$. The upper dashed curve is for both qubits subject to noise while the lower solid one is for the case where a qubit is isolated.}
\label{fig:purity}
\end{figure}

\section{Conclusion}
In this paper we have investigated the possibility to maintain, during a pure-dephasing evolution, an amount entanglement capable to ensure the presence of nonlocal correlations within a two-qubit system. To this purpose, local pure-dephasing environments characterized by a Ohmic-like spectral density have been taken into account. We have reviewed a known closed formula which dynamically relates the Bell function $\mathcal{B}$, witnessing nonlocality, and the concurrence $\mathcal{C}$, quantifying entanglement, for arbitrary local pure-dephasing evolutions \cite{lofrancoPRB}, giving a formal demonstration of it. This relation, which permits to find a threshold $\mathcal{C}_\mathrm{th}$ for the concurrence such that for any value of $\mathcal{C}>\mathcal{C}_\mathrm{th}$ the system state exhibits nonlocality, has been then exploited within the model with Ohmic-like spectrum. 

We have shown that stationary entanglement is achievable for a given degree of ohmicity $s$ and for a suitable purity $r$ of the initial state. In particular, we have found the analytical expression of the lower bound of initial purity $r^\ast$ such that, for a given $s$, the stationary entanglement guarantees nonlocality for any value of purity of the initial state $r>r^\ast$.   
We remark that, although the most performing condition for entanglement trapping happens for a value of $s$, namely $\bar{s}\simeq2.46$, which enables a non-Markovian regime with the occurrence of local information backflows from the environment to the system \cite{addisPRA}, the achievement of stationary entanglement is simply related to the existence of super-Ohmic environments ($s>1$) for which the Markovian approximation fails and non-Markovian effects arise even without information backflows \cite{addisPRA}.    
We also point out that stationary entanglement was obtained for two noninteracting qubits locally subject to amplitude damping channels each with a super-Ohmic spectral density \cite{non-Mar2}, where the mechanism leading to this phenomenon is linked to the presence of a bound state between each qubit with its local dissipative bosonic environment. Another recent cavity-based architecture has been reported which is able in principle to give stationary entanglement in a dissipative compound environment configuration \cite{man2015}. In the system considered in this paper, instead, the pure-dephasing interaction is non-dissipative and no qubit-environment bound state can be created. The origin of the stationary entanglement, activating nonlocality, is here basically due to the manifestation of coherence trapping for the localized qubits.

The results presented in this work highlight the importance of knowing closed relations between different kinds of quantum correlations, as entanglement and nonlocality. Indeed this knowledge allows for finding preparation conditions of the initial states, under opportune environmental characteristics, which ensure that a given kind of quantum correlations is preserved in the long-time limit, for instance nonlocal correlations, just by maintaining another type of correlations, e.g. quantum entanglement, above a certain amount. Such findings can count not only for fundamental studies but also for practical purposes in quantum communication and information processing.

\begin{acknowledgements}
R.L.F. would like to acknowledge Gerardo Adesso and his group at the University of Nottingham for the very kind hospitality during the completion of this work. R.L.F. also thanks Sigur R\'{o}s for fruitful meetings. 
\end{acknowledgements}


\begin{thebibliography}{10}
\expandafter\ifx\csname url\endcsname\relax
  \def\url#1{\texttt{#1}}\fi
\expandafter\ifx\csname urlprefix\endcsname\relax\def\urlprefix{URL }\fi
\providecommand{\bibinfo}[2]{#2}
\providecommand{\eprint}[2][]{\url{#2}}

\bibitem{amico2008RMP}
\bibinfo{author}{L. Amico}, \bibinfo{author}{R. Fazio},
  \bibinfo{author}{A. Osterloh} and \bibinfo{author}{V. Vedral}:
\newblock \bibinfo{title}{Entanglement in many-body systems},
\newblock \emph{\bibinfo{journal}{Rev. Mod. Phys.}}
  \textbf{\bibinfo{volume}{80}}, \bibinfo{pages}{517–576}
  (\bibinfo{year}{2008}).

\bibitem{horodecki2009RMP}
\bibinfo{author}{R. Horodecki}, \bibinfo{author}{P. Horodecki},
  \bibinfo{author}{M. Horodecki} and \bibinfo{author}{K. Horodecki}:
\newblock \bibinfo{title}{Quantum entanglement},
\newblock \emph{\bibinfo{journal}{Rev. Mod. Phys.}}
  \textbf{\bibinfo{volume}{81}}, \bibinfo{pages}{865--942}
  (\bibinfo{year}{2009}).

\bibitem{brunnerRMP}
\bibinfo{author}{N. Brunner}, \bibinfo{author}{D. Cavalcanti},
  \bibinfo{author}{S. Pironio}, \bibinfo{author}{V. Scarani} and
  \bibinfo{author}{S. Wehner}:
\newblock \bibinfo{title}{Bell nonlocality},
\newblock \emph{\bibinfo{journal}{Rev. Mod. Phys.}}
  \textbf{\bibinfo{volume}{86}}, \bibinfo{pages}{419} (\bibinfo{year}{2014}).

\bibitem{werner}
\bibinfo{author}{R. F. Werner}:
\newblock \bibinfo{title}{Quantum states with {Einstein-Podolsky-Rosen}
  correlations admitting a hidden-variable model},
\newblock \emph{\bibinfo{journal}{Phys. Rev. A}} \textbf{\bibinfo{volume}{40}},
  \bibinfo{pages}{4277} (\bibinfo{year}{1989}).

\bibitem{petru}
\bibinfo{author}{H.-P. Breuer} and \bibinfo{author}{F. Petruccione}:
\newblock \emph{\bibinfo{title}{The Theory of Open Quantum Systems}},
 \bibinfo{publisher}{Oxford University Press}, \bibinfo{address}{Oxford, New
  York}, \bibinfo{year}{2002}.

\bibitem{rivasreview}
\bibinfo{author}{{\'{A}}. Rivas}, \bibinfo{author}{S.~F. Huelga} and
  \bibinfo{author}{M.~B. Plenio}:
\newblock \bibinfo{title}{{Quantum non-Markovianity: characterization,
  quantification and detection}},
\newblock \emph{\bibinfo{journal}{Rep. Prog. Phys.}}
  \textbf{\bibinfo{volume}{77}}, \bibinfo{pages}{094001}
  (\bibinfo{year}{2014}).

\bibitem{lofrancoreview}
\bibinfo{author}{R. {Lo Franco}}, \bibinfo{author}{B. Bellomo},
  \bibinfo{author}{S. Maniscalco} and \bibinfo{author}{G. Compagno}:
\newblock \bibinfo{title}{Dynamics of quantum correlations in two-qubit systems
  within non-{Markovian} environments},
\newblock \emph{\bibinfo{journal}{Int. J. Mod. Phys. B}}
  \textbf{\bibinfo{volume}{27}}, \bibinfo{pages}{1345053}
  (\bibinfo{year}{2013}).
  
  \bibitem{aolitareview}
\bibinfo{author}{L. Aolita}, \bibinfo{author}{F. {de Melo}} and \bibinfo{author}{L. Davidovich}:
\newblock \bibinfo{title}{Open-system dynamics of entanglement: a key issues review},
\newblock \emph{\bibinfo{journal}{Rep. Prog. Phys.}}
  \textbf{\bibinfo{volume}{78}}, \bibinfo{pages}{042001}
  (\bibinfo{year}{2015}).
  
\bibitem{obrienreview}
\bibinfo{author}{T. D. Ladd} \emph{et~al.}:
\newblock \bibinfo{title}{Quantum computers},
\newblock \emph{\bibinfo{journal}{Nature}} \textbf{\bibinfo{volume}{464}},
  \bibinfo{pages}{45} (\bibinfo{year}{2010}).
  
\bibitem{cirac2001PRA}
\bibinfo{author}{J. Schliemann}, \bibinfo{author}{J. I. Cirac},
  \bibinfo{author}{M. Kus}, \bibinfo{author}{M. Lewenstein} and \bibinfo{author}{D. Loss}:
\newblock \bibinfo{title}{Quantum correlations in two-fermion systems},
\newblock \emph{\bibinfo{journal}{Phys. Rev. A}}
  \textbf{\bibinfo{volume}{64}}, \bibinfo{pages}{022303}
  (\bibinfo{year}{2001}).  
  
\bibitem{xu2013}
J.-S. Xu \emph{et al.}: Experimental recovery of quantum correlations in absence of system-environment back-action,
\emph{Nat. Commun.} \textbf{4}, 2851 (2013).

\bibitem{darrigoAOP}
A. {D'Arrigo}, R. Lo Franco, G. Benenti, E. Paladino and G. Falci: Recovering entanglement by local operations,
\emph{Ann. Phys.} \textbf{350}, 211 (2015).

\bibitem{orieux2015}
A. Orieux \emph{et al.}: Experimental on-demand recovery of quantum entanglement by local operations within non-Markovian dynamics,
\emph{Sci. Rep.} \textbf{5}, 8575 (2015).

\bibitem{mazzolapalermo2010PRA}
\bibinfo{author}{L. Mazzola}, \bibinfo{author}{B. Bellomo},
  \bibinfo{author}{{R. Lo Franco}} and \bibinfo{author}{G. Compagno}:
\newblock \bibinfo{title}{Connection among entanglement, mixedness and
  nonlocality in a dynamical context},
\newblock \emph{\bibinfo{journal}{Phys. Rev. A}} \textbf{\bibinfo{volume}{81}},
  \bibinfo{pages}{052116} (\bibinfo{year}{2010}).

\bibitem{horst}
\bibinfo{author}{B. Horst}, \bibinfo{author}{K. Bartkiewicz} and
  \bibinfo{author}{A. Miranowicz}:
\newblock \bibinfo{title}{{Two-qubit mixed states more entangled than pure
  states: Comparison of the relative entropy of entanglement for a given
  nonlocality}},
\newblock \emph{\bibinfo{journal}{Phys. Rev. A}} \textbf{\bibinfo{volume}{87}},
  \bibinfo{pages}{042108} (\bibinfo{year}{2013}).

\bibitem{Bartkiewicz}
\bibinfo{author}{K. Bartkiewicz}, \bibinfo{author}{B. Horst},
  \bibinfo{author}{K. Lemr} and \bibinfo{author}{A. Miranowicz}:
\newblock \bibinfo{title}{{Entanglement estimation from Bell inequality
  violation}},
\newblock \emph{\bibinfo{journal}{Phys. Rev. A}} \textbf{\bibinfo{volume}{88}},
  \bibinfo{pages}{052105} (\bibinfo{year}{2013}).

\bibitem{Wootters98}
\bibinfo{author}{W. K. Wootters}:
\newblock \bibinfo{title}{Entanglement of formation of an arbitrary state of
  two qubits},
\newblock \emph{\bibinfo{journal}{Phys. Rev. Lett.}}
  \textbf{\bibinfo{volume}{80}}, \bibinfo{pages}{2245--2248}
  (\bibinfo{year}{1998}).

\bibitem{gisinnatphot}
\bibinfo{author}{N. Gisin} and \bibinfo{author}{R. Thew}:
\newblock \bibinfo{title}{Quantum communication},
\newblock \emph{\bibinfo{journal}{Nature Photon.}}
  \textbf{\bibinfo{volume}{1}}, \bibinfo{pages}{165} (\bibinfo{year}{2007}).

\bibitem{lofrancoPRB}
\bibinfo{author}{{R. Lo Franco}}, \bibinfo{author}{A. D'Arrigo},
  \bibinfo{author}{G. Falci}, \bibinfo{author}{G. Compagno} and
  \bibinfo{author}{E. Paladino}:
\newblock \bibinfo{title}{Preserving entanglement and nonlocality in
  solid-state qubits by dynamical decoupling},
\newblock \emph{\bibinfo{journal}{Phys. Rev. B}} \textbf{\bibinfo{volume}{90}},
  \bibinfo{pages}{054304} (\bibinfo{year}{2014}).

\bibitem{addisPRA}
\bibinfo{author}{C. Addis}, \bibinfo{author}{G. Brebner},
  \bibinfo{author}{P. Haikka} and \bibinfo{author}{S. Maniscalco}:
\newblock \bibinfo{title}{{Coherence trapping and information backflow in
  dephasing qubits}},
\newblock \emph{\bibinfo{journal}{Phys. Rev. A}} \textbf{\bibinfo{volume}{89}},
  \bibinfo{pages}{024101} (\bibinfo{year}{2014}).

\bibitem{bellomo2010PLA}
\bibinfo{author}{B. Bellomo}, \bibinfo{author}{{R. Lo Franco}} and
  \bibinfo{author}{G. Compagno}:
\newblock \bibinfo{title}{An optimized Bell test in a dynamical system},
\newblock \emph{\bibinfo{journal}{Phys. Lett. A}}
  \textbf{\bibinfo{volume}{374}}, \bibinfo{pages}{3007} (\bibinfo{year}{2010}).

\bibitem{horodecki1995PLA}
\bibinfo{author}{M. Horodecki}, \bibinfo{author}{P. Horodecki} and
  \bibinfo{author}{R. Horodecki}:
\newblock \bibinfo{title}{{Violating Bell inequality by mixed spin-{1}/{2}
  states: necessary and sufficient condition}},
\newblock \emph{\bibinfo{journal}{Phys. Lett. A}}
  \textbf{\bibinfo{volume}{200}}, \bibinfo{pages}{340} (\bibinfo{year}{1995}).

\bibitem{bellomo2008PRA}
\bibinfo{author}{B. Bellomo}, \bibinfo{author}{{R. Lo Franco}} and
  \bibinfo{author}{G. Compagno}:
\newblock \bibinfo{title}{Entanglement dynamics of two independent qubits in
  environments with and without memory},
\newblock \emph{\bibinfo{journal}{Phys. Rev. A}} \textbf{\bibinfo{volume}{77}},
  \bibinfo{pages}{032342} (\bibinfo{year}{2008}).

\bibitem{bellomo2007PRL}
\bibinfo{author}{B. Bellomo}, \bibinfo{author}{{R. Lo Franco}} and
  \bibinfo{author}{G. Compagno}:
\newblock \bibinfo{title}{Non-{Markovian} effects on the dynamics of
  entanglement},
\newblock \emph{\bibinfo{journal}{Phys. Rev. Lett.}}
  \textbf{\bibinfo{volume}{99}}, \bibinfo{pages}{160502}
  (\bibinfo{year}{2007}).

\bibitem{derkacz2005PRA}
\bibinfo{author}{L. Derkacz} and \bibinfo{author}{L. Jak{\'{o}}bczyk}:
\newblock \bibinfo{title}{{Clauser-Horne-Shimony-Holt violation and the
  entropy-concurrence plane}},
\newblock \emph{\bibinfo{journal}{Phys. Rev. A}} \textbf{\bibinfo{volume}{72}},
  \bibinfo{pages}{042321} (\bibinfo{year}{2005}).

\bibitem{gisin1991PLA}
\bibinfo{author}{N. Gisin}:
\newblock \bibinfo{title}{{Bell's inequality holds for all non-product
  states}},
\newblock \emph{\bibinfo{journal}{Phys. Lett. A}}
  \textbf{\bibinfo{volume}{154}}, \bibinfo{pages}{201} (\bibinfo{year}{1991}).

\bibitem{verstraete2002PRL}
\bibinfo{author}{F. Verstraete} and \bibinfo{author}{M.~M. Wolf}:
\newblock \bibinfo{title}{{Entanglement versus Bell violations and their
  behavior under local filtering operations}},
\newblock \emph{\bibinfo{journal}{Phys. Rev. Lett.}}
  \textbf{\bibinfo{volume}{89}}, \bibinfo{pages}{170401}
  (\bibinfo{year}{2002}).

\bibitem{haikkaPRA}
\bibinfo{author}{P. Haikka}, \bibinfo{author}{T.~H. Johnson} and
  \bibinfo{author}{S. Maniscalco}:
\newblock \bibinfo{title}{Non-Markovianity of local dephasing channels and
  time-invariant discord},
\newblock \emph{\bibinfo{journal}{Phys. Rev. A}} \textbf{\bibinfo{volume}{87}},
  \bibinfo{pages}{010103(R)} (\bibinfo{year}{2013}).

\bibitem{cianciarusoJPA}
\bibinfo{author}{T.~R. Bromley}, \bibinfo{author}{M. Cianciaruso},
  \bibinfo{author}{{R. Lo Franco}} and \bibinfo{author}{G. Adesso}:
\newblock \bibinfo{title}{{Unifying approach to the quantification of bipartite
  correlations by Bures distance}},
\newblock \emph{\bibinfo{journal}{J. Phys. A: Math. Theor.}}
  \textbf{\bibinfo{volume}{47}}, \bibinfo{pages}{405302}
  (\bibinfo{year}{2014}).
  
  \bibitem{chiuri2012}
\bibinfo{author}{A. Chiuri}, \bibinfo{author}{C. Greganti},
  \bibinfo{author}{L. Mazzola},  \bibinfo{author}{M. Paternostro} and \bibinfo{author}{P. Mataloni}:
\newblock \bibinfo{title}{{Linear optics simulation of quantum non-Markovian dynamics}},
\newblock \emph{\bibinfo{journal}{Sci. Rep.}}
  \textbf{\bibinfo{volume}{2}}, \bibinfo{pages}{968}
  (\bibinfo{year}{2012}).
  
    \bibitem{pfaff2013}
\bibinfo{author}{W. Pfaff}, \bibinfo{author}{T. H. Taminiau},
  \bibinfo{author}{L. Robledo}, \bibinfo{author}{H. Bernien}, \bibinfo{author}{M. Markham}, \bibinfo{author}{D. J. Twitchen} and \bibinfo{R. Hanson}:
\newblock \bibinfo{title}{Demonstration of entanglement-by-measurement of solid-state qubits},
\newblock \emph{\bibinfo{journal}{Nat. Phys.}}
  \textbf{\bibinfo{volume}{9}}, \bibinfo{pages}{29}
  (\bibinfo{year}{2013}).
  
      \bibitem{metwally2008}
\bibinfo{author}{N. Metwally}:
\newblock \bibinfo{title}{New aspects of the purity and information of an entangled qubit pair},
\newblock \emph{\bibinfo{journal}{Int. J. Quantum Inform.}}
  \textbf{\bibinfo{volume}{6}}, \bibinfo{pages}{187}
  (\bibinfo{year}{2008}).

\bibitem{non-Mar2}
\bibinfo{author}{J. Tan}, \bibinfo{author}{T.~H. Kyaw} and
  \bibinfo{author}{Y. Yeo}:
\newblock \bibinfo{title}{Non-{Markovian} environments and entanglement
  preservation},
\newblock \emph{\bibinfo{journal}{Phys. Rev. A}} \textbf{\bibinfo{volume}{81}},
  \bibinfo{pages}{062119} (\bibinfo{year}{2010}).
  
\bibitem{man2015}
\bibinfo{author}{Z.-X. Man}, \bibinfo{author}{Y.-J. Xia} and
  \bibinfo{author}{R. {Lo Franco}}:
\newblock \bibinfo{title}{Cavity-based architecture to preserve quantum coherence and entanglement},
\newblock \emph{\bibinfo{journal}{Sci. Rep.}} \textbf{\bibinfo{volume}{5}},
  \bibinfo{pages}{13843} (\bibinfo{year}{2015}).

\end{thebibliography}
\end{document}